\documentclass[a4paper,11pt]{amsart}

\usepackage{a4wide, amsmath, amsfonts, amssymb, mathrsfs, amsthm, bbm, color, stmaryrd, breqn}
\usepackage[hyperfootnotes=false]{hyperref}

\usepackage[T1]{fontenc}
\usepackage[latin9]{inputenc}
\usepackage{color}
\usepackage{amsmath}
\usepackage{amssymb}
\usepackage{a4wide, amsmath, amsfonts, amssymb, mathrsfs, amsthm, bbm, color, stmaryrd, breqn}
\usepackage{esint, mathrsfs, amsthm, bbm, color, stmaryrd, breqn}

\makeatletter

\begin{document}

\global\long\def\intr{\int_{R}}
 \global\long\def\sbr#1{\left[ #1\right] }
 \global\long\def\cbr#1{\left\{  #1\right\}  }
 \global\long\def\rbr#1{\left(#1\right)}
 \global\long\def\ev#1{\mathbb{E}{#1}}
 \global\long\def\E{\mathbb{E}}
 \global\long\def\P{\mathbb{P}}
 \global\long\def\R{\mathbb{R}}
 \global\long\def\N{\mathbb{N}}
  \global\long\def\Z{\mathbb{Z}}
 \global\long\def\norm#1#2#3{\Vert#1\Vert_{#2}^{#3}}
 \global\long\def\pr#1{\mathbb{P}\rbr{#1}}
 \global\long\def\cleq{\lesssim}
 \global\long\def\ceq{\eqsim}
 \global\long\def\conv{\rightarrow}
 \global\long\def\Var#1{\text{Var}(#1)}
 \global\long\def\TDD{}
 \global\long\def\dd#1{\textnormal{d}#1}
 \global\long\def\inti{\int_{0}^{\infty}}
 \global\long\def\crr{\mathcal{C}([0;\infty),\R)}
 \global\long\def\sb#1{\langle#1\rangle}
 \global\long\def\pm#1{d_{P}\rbr{#1}}
 \global\long\def\crt{\mathcal{C}([0;T],\R)}
 \global\long\def\nuu{\nu_{n;\lambda}}
 \global\long\def\ZZ{Z_{\Lambda_{n}}}
 \global\long\def\PP{\mathbb{P}_{\Lambda_{n}}}
 \global\long\def\EE{\mathbb{E}_{\Lambda_{n}}}
 \global\long\def\LL{\Lambda_{n}}
 \global\long\def\AA{\mathcal{A}}
 \global\long\def\evx{\mathbb{E}_{x}}
 \global\long\def\pin#1{1_{\cbr{#1\in\mathcal{A}}}}
 \global\long\def\Zd{\mathbb{Z}^{d}}
 \global\long\def\ddp#1#2{\langle#1,#2\rangle}
 \global\long\def\intc#1{\int_{0}^{#1}}
 \global\long\def\T#1{\mathcal{P}_{#1}}
 \global\long\def\ii{\mathbf{i}}
 \global\long\def\star#1{\left.#1^{*}\right.}
 \global\long\def\pspace{\mathcal{C}}
 \global\long\def\eq{\varphi}
 \global\long\def\grad{\text{grad}}
 \global\long\def\var{\text{var}}
 \global\long\def\ab{[a,b]}
 \global\long\def\ra{\rightarrow}
 \global\long\def\TTV#1#2#3{\text{TV}^{#3}\!\rbr{#1,#2}}
 \global\long\def\TTVemph#1#2#3{\emph{TV}^{#3}\!\rbr{#1,#2}}
 \global\long\def\V#1#2#3{\text{V}^{#3}\!\rbr{#1,#2}}
 \global\long\def\vi#1#2{\text{vi}\rbr{#1,#2}}
 \global\long\def\eqdef{:=}
 \global\long\def\UTV#1#2#3{\text{UTV}^{#3}\!\rbr{#1,#2}}
 \global\long\def\DTV#1#2#3{\text{DTV}^{#3}\!\rbr{#1,#2}}
 \global\long\def\ns{\infty}
 \global\long\def\f{:\left[a,b\right]\ra\R}
 \global\long\def\TV{\text{TV}}
 \global\long\def\osc{\text{osc}}
 \global\long\def\mathcalu{{\mathcal U}^{p}\left[a,b\right]}
 \global\long\def\cont{\text{cont}}
 \global\long\def\oP{\mathbb{\bar{P}}}
 \global\long\def\oP{\overline{\mathbb{P}}}

\newtheorem{theorem}{Theorem} \newtheorem{tw}[theorem]{Theorem}
\newtheorem{stw}[theorem]{Proposition} \newtheorem{lem}[theorem]{Lemma}
\newtheorem{rem}[theorem]{Remark} \newtheorem{defi}[theorem]{Definition}
\newtheorem{definition}[theorem]{Definition} \newtheorem{corollary}[theorem]{Corollary}

\title{On the quadratic variation of the model-free price paths with jumps}

\author[Galane]{Lesiba Ch. Galane}
\address{Lesiba Ch. Galane, Department of Mathematics and Applied Mathematics, University of Limpopo, South Africa}
\email{Lesiba.Galane@ul.ac.za}

\author[\L ochowski]{Rafa\l{} M. {}\L ochowski}
\address{Rafa\l{} M. \L ochowski, Department of Mathematics and Mathematical Economics, Warsaw School of Economics, Poland, and African Institute for Mathematical Sciences, South Africa}
\email{rlocho314@gmail.com}

\author[Mhlanga]{Farai J. Mhlanga}
\address{Farai J. Mhlanga, Department of Mathematics and Applied Mathematics, University of Limpopo, South Africa}
\email{Farai.Mhlanga@ul.ac.za}

\maketitle

\begin{abstract}
We prove that the model-free typical (in the sense of Vovk) c\`adl\`ag price paths with mildly restricted downward jumps possess quadratic variation which does not depend on the specific sequence of partitions as long as these partitions are obtained from stopping times such that the oscillations of a path on the consecutive (half-open on the right)
intervals of these partitions tend (in a specified sense) to 0. Finally, we also define quasi-explicit, partition independent quantities which tend to this quadratic variation.
\end{abstract}

\section{Introduction}

In this paper we deal with several properties of the quadratic variation
of model-free, c\`{a}dl\`{a}g price paths and integrals driven by such paths.
In \cite{Vovk_cadlag:2015} Vovk proved the existence of quadratic
variation of c\`{a}dl\`{a}g price paths with mildly restricted jumps (along
the sequence of so called Lebesgue partitions). In \cite{LochPerkPro:2016}
this result was generalised and the existence of quadratic variation along the same sequence of Lebesgue partitions of
c\`{a}dl\`{a}g price paths with mildly restricted downward jumps was proven (this in particular proves the existence of quadratic variation of non-negative
c\`{a}dl\`{a}g price paths).
Having these results in hand and a pathwise version of the BDG inequality
proven in \cite{Beiglboeck:2015}, it became possible to define a
stochastic integral along c\`{a}dl\`{a}g price paths with mildly restricted downward jumps for broad class of integrands.
Moreover, there were also proven some continuity results for such
integrals. Other approach for integrators with jumps whose absolute
or relative size is bounded by some constant, was presented in \cite{Vovk_integral:2015}.
In \cite{Vovk_integral:2015} (see also \cite{Vovk_Schafer_semimart:2016}) there was also introduced a very interesting notion of \emph{uniform on compacts quasi-always (ucqa)} convergence, which roughly means that the trader becomes infinitely rich immediately after the convergence ceases to hold.
Unfortunately we were able only to prove rather weak modes of convergence, for which continuity results (like Theorem \ref{thm:integral} and Corollary \ref{cor:c\`{a}dl\`{a}g concentration estimate}) were available.

Since the existence of quadratic variation of price paths in model-free
finance is of utmost importance (also in practical sense, as it
corresponds to the well known notion of realized volatility), it is
problematic that this fundamental object \emph{a priori} depends on
the choice of partitions. Fortunately, using the mentioned continuity
results for model-free integrals, in this note we prove the independence
of the quadratic variation of model-free price paths as long as the
partitions are obtained from \emph{stopping times}. This result is
analogous to the results obtained for continuous (model-free) price
paths by Vovk and Shafer in \cite{Vovk_Schafer_semimart:2016}, for
continuous semimartingales by Davis, Ob\l\'{o}j and Siorpaes in \cite{Obloj_local:2015}
and for c\`{a}dl\`{a}g semimartingales in Cont and Fourni$\acute{\text{e}}$ \cite{ContFournie:2010}.
Further, we prove another, partition-independent formula for the quadratic
variation of model-free price paths with jumps in terms of the \emph{truncated
variation}. Finally, we deal with possibility of the extension of
this formula for deterministic c\`{a}dl\`{a}g functions possessing quadratic
variation along some sequence of (refining) partitions in the sense
of F$\ddot{\text{o}}$llmer. We start with a technical result, which will be used in the sequel --- integration by parts formula for model-free price paths perturbed by an adapted, finite
variation path.

\section{Notation, definitions and continuity of model-free, pathwise integrals}

Let $d\in\left\{ 1,2,\ldots\right\} ,$ $T\in\left(0,+\ns\right),$
$\N=\left\{ 1,2,\ldots\right\}$, $\N_0 = \N \cup \cbr{0}$, $\Z$ be the set of all integers and $\R_{+}=\left[0,+\ns\right).$
For a topological space ${\mathcal E}$ we will denote by $D([0,T],{\mathcal E})$ the space of all c\`{a}dl\`{a}g functions $\omega:\left[0,T\right]\ra{\mathcal E}.$
Now let $\psi:\R_{+}\ra\R_{+}$ be some
non-decreasing function.
We equip $\R^d$ with ${l}^2$ norm $|\cdot|$ and consider the underlying space $\Omega$ which is a subset of the set $\Omega_{\psi}\subset D\rbr{[0, T], \R^d}$ of c\`adl\`ag functions with mildly restricted jumps directed downwards, that is $\omega=(\omega^1,\dots,\omega^d) \in \Omega_\psi$ if for $i=1,\ldots , d$  it satisfies
\begin{equation*}
  \Delta\omega^i (t) := \omega^i (t) -\omega^i (t-)\geq -\psi \bigg(\sup_{s\in[0,t)}|\omega(s)|\bigg),\quad t\in (0,T],
\end{equation*}
where $\omega^i (t-):= \lim_{s\to t, \, s<t} \omega (s)$.

The following sample spaces are examples of~$\Omega$:
  \begin{enumerate}
    \item $\Omega_c:=C([0,T],\R^d)$, the space of all continuous functions $\omega\colon [0,T]\to\R^d$,
    \item $\Omega_+ := D([0,T], \R^d_+)$, the space of all non-negative c\`adl\`ag functions $\omega\colon [0,T]\to\R^d_+$ (here $\psi(x)=x$),
    \item $\tilde \Omega_{\psi}$ which is defined as the subset of all c\`adl\`ag functions $\omega\colon [0,T]\to\R^d$ such that
          \begin{equation*}
            |\omega (t)-\omega(t-)|\leq \psi\bigg(\sup_{s\in[0,t)}|\omega(s)|\bigg),\quad t\in (0,T],
          \end{equation*}
          and $\psi\colon\R_+\to (0,\infty)$ is a fixed non-decreasing function.
   \end{enumerate}
   A detailed financial interpretation of the last space can be found in~\cite{Vovk_cadlag:2015} and a generalization of this space allowing for different bounds for jumps directed upwards resp. downwards was recently introduced in~\cite{Vovk_integral:2015}.

$\Omega$ will be our sample space and for each
$t\in\left[0,T\right],$ ${\mathcal F}_{t}^{\circ}$ is defined to be
the smallest $\sigma$-algebra on $\Omega$ that makes all functions
$\omega\mapsto\omega\left(s\right),$ $s\in\left[0,t\right],$ measurable
and ${\mathcal F}_{t}$ is defined to be the universal completion of ${\mathcal F}_{t}^{\circ}.$
An event is an element of the $\sigma$-algebra ${\mathcal F}_{T}.$ Stopping
times $\tau:\Omega\ra\left[0,T\right]\cup\left\{ \ns\right\} $ with
respect to the filtration $\left({\mathcal F}_{t}\right)_{t\in\left[0,T\right]}$
and the corresponding $\sigma$-algebras ${\mathcal F}_{\tau}$ are defined
as usual. $S$ is the coordinate process, i.e. $S_{t}\left(\omega\right):=\omega\left(t\right).$

Now a process $H:\Omega\times\left[0,T\right]\ra\R^{d}$ is called
a \emph{simple strategy} if there exist stopping times $0=\tau_{0}\le\tau_{1}\le\ldots$
and ${\mathcal F}_{\tau_{k}}$-measurable bounded functions $h_{k}:\Omega\ra\R^{d},$ $k\in \N_0$,
such that for each $\omega\in\Omega,$ $\tau_{K}\left(\omega\right)=\tau_{K+1}\left(\omega\right)=\ldots\in \left[0,T\right]\cup\left\{ \ns\right\}$
from some $K\in \N_0$ on, and
\[
H_{t}\left(\omega\right)=h_{0}\left(\omega\right){\bf 1}_{\{0\}}\left(t\right)+\sum_{k=0}^{+\ns}h_{k}\left(\omega\right){\bf 1}_{\left(\tau_{k}\left(\omega\right),\tau_{k+1}\left(\omega\right)\right]}\left(t\right).
\]
In what follows, when $\tau$ is a stopping time and $\omega \in \Omega$ we will often write $\tau$ instead of $\tau(\omega)$.

For the simple strategy $H$ the corresponding \emph{integral process}
\begin{equation}
(H\cdot S)_{t}(\omega):=\sum_{k=0}^{\infty}h_{k}(\omega)\cdot(S_{\tau_{k+1}\wedge t}(\omega)-S_{\tau_{k}\wedge t}(\omega))=\sum_{n=0}^{\infty}h_{k}(\omega) \cdot S_{\tau_{k}\wedge t,\tau_{k+1}\wedge t}(\omega) \label{eq:step function integral}
\end{equation}
is well-defined for all $\omega\in\Omega$ and all $t\in[0,T]$; here ''$\cdot$'' denotes the scalar product on $\R^d$ and $S_{u,v}:=S_{v}-S_{u}$ for $u,v\in[0,T]$.
The family of simple strategies will be denoted by ${\mathcal H}.$ For
$\lambda>0$ a simple strategy $H$ will be called \emph{(strongly)
$\lambda$-admissible} if $(H\cdot S)_{t}(\omega)\ge-\lambda$ for
all $\omega\in\Omega$ and all $t\in[0,T]$. The set of strongly $\lambda$-admissible
simple strategies will be denoted by $\mathcal{H}_{\lambda}$.

\begin{definition} \emph{Vovk's outer measure} $\oP$ of a set
$A\subseteq\Omega$ is defined as the minimal superhedging price for
$\mathbf{1}_{A}$, that is
\[
\oP(A):=\inf\Big\{\lambda>0\,:\,\exists(H^{n})_{n\in\N}\subseteq\mathcal{H}_{\lambda}\text{ s.t. }\forall\omega\in\Omega\,:\,\liminf_{n\to\infty}(\lambda+(H^{n}\cdot S)_{T}(\omega))\ge\mathbf{1}_{A}(\omega)\Big\}.
\]
A set $A\subseteq\Omega$ is called a \emph{null set} if it has outer
measure zero. A property (P) holds for \emph{typical price paths}
if the set $A$ where (P) is violated is a null set. \end{definition}

This definition differs slightly from original Vovk's definition (see
for example \cite{Vovk_probability:2012}). A detailed account on the difference
of the original Vovk's outer measure and the just defined measure is presented
in \cite[Subsection 2.3]{PerkowskiProemel_integral:2015}. However,
since we will need here some continuity results of the model-free
integrals which were obtained for the just defined outer measure,
we will use this measure instead of original Vovk's outer measure.

Let us now introduce the notion of $n$th Lebesque partition $\pi^n$, $n \in \N$, which we will use in the sequel.   For $n\in \N$ we define  $\mathbb{D}^n:=\{k2^{-n}\,:\, k\in\mathbb{Z}\}$. For $f\in D\rbr{[0,T], \R}$, $\pi^n(f)$~consists of points in time at which the  path~$f$ crosses  (in space) a dyadic number from~$\mathbb{D}^n$ which is not the same as the dyadic number crossed (as the last number from~$\mathbb{D}^n$) at the  preceding time. This idea is made precise in the following definition.
\begin{definition}\label{def:Lebesgue partition}
  Let $n\in\mathbb{N}$ and let $\mathbb{D}^n$ be the {$n$} generation of dyadic numbers. For $ f\in D\rbr{[0,T], \R}$ its \emph{$n$th Lebesgue partition} $\pi^n(f):=\{\pi^n_k (f)\,:\, k \in \N_0 \}$ is given by the sequence of times $(\pi^n_k(f))_{k\in \mathbb{N}_0}$ inductively defined by
  \begin{equation*}
    \pi^n_{0}(f):=0\quad \text{and}\quad D^n_0(f):=\max \{\mathbb{D}^n\cap(-\infty,S_0(f)]\},
  \end{equation*}
and for every $k\in \mathbb{N}$ we further set
  \begin{align*}
    \pi^n_k(f) &:= \inf \{ t\in[\pi^n_{k-1}(f),T]\,:\, \llbracket f \rbr{\pi^n_{k-1}(f) }, f (t) \rrbracket \cap (\mathbb{D}^n\setminus \cbr{D^n_{k-1}(f)}) \ne \emptyset \},\\
    D^n_k(f) &:= \emph{argmin}_{D \in \llbracket f \rbr{\pi^n_{k-1}(f) }, f \rbr{\pi^n_{k}(f) } \rrbracket \cap (\mathbb{D}^n\setminus \cbr{D^n_{k-1}(f)})} |D - f \rbr{\pi^n_{k}(f)}|
     \end{align*}
  with the convention $\inf \emptyset = \infty$ and $ \llbracket u,v\rrbracket := [\min\rbr{u,v}, \max\rbr{u,v}]$.
\end{definition}
Next we define the sequence of Lebesgue partitions for $d$-dimensional c\`adl\`ag function $\omega\in \Omega$.
\begin{definition}\label{def:multi Lebesgue partition}
For $n\in \N$ and $\omega \in D\rbr{[0,T], \R^d}$ its \emph{Lebesgue partition} $\pi^n(\omega):=\{\pi^n_k (\omega)\,:\, k \in \N_0 \}$ is iteratively defined by $\pi^n_0(\omega):=0$  and
  \begin{equation*}
    \pi^n_k(\omega) := \min\bigg \{ t > \pi^n_{k-1}(\omega)\,:\, t \in \bigcup_{i=1}^d \pi^n(\omega^i)\cup  \bigcup_{i,j=1,i\neq j}^d \pi^n(\omega^i+\omega^j) \bigg\},\quad k\in \N,
  \end{equation*}
  where $\omega = (\omega^1,\dots,\omega^d)$ and $\pi^n(\omega^i)$ and $\pi^n(\omega^i +\omega^j)$ are the Lebesgue partitions of $\omega^i$ and $\omega^i+\omega^j$ as introduced in Definition~\ref{def:Lebesgue partition}, respectively.
\end{definition}
By \cite[Lemma 3]{Vovk_cadlag:2015} we know that for $n \in \N$, $\pi^n$ defined in Definition~\ref{def:multi Lebesgue partition} is indeed a family of stopping times and by \cite[Corollary 3.11]{LochPerkPro:2016} (and \cite{Vovk_cadlag:2015} for $\Omega=\tilde \Omega_{\psi}$) we know that
for typical price path $\omega\in\Omega$ the sequence of discrete quadratic (co)variations along the sequence of Lebesgue partitions:
\begin{equation} \label{discrete_qv}
Q_{t}^{i,j,n}(\omega):=\sum_{k=1}^{\infty}S_{\pi_{k-1}^{n}\wedge t, \pi_{k}^{n}\wedge t}^{i}(\omega)S_{\pi_{k-1}^{n}\wedge t, \pi_{k}^{n}\wedge t}^{j}(\omega),\quad t\in[0,T],
\end{equation}
converges for $i,j =1,2,\ldots, d$ in the uniform topology to some (c\`{a}dl\`{a}g) function $[0, T] \ni t \mapsto [S^{i},S^{j}]_t(\omega)$ which we will call the quadratic (co)variation of $S^i$ and $S^j$.

 We will use the following notation: $\left|[S]_{T}(\omega)\right|=\rbr{\sum_{i,j=1}^{d}[S^{i},S^{j}]_{T}^{2}(\omega)}^{1/2}.$

For $Z:\Omega\times\left[0,T\right]\ra\R^{r}$ ($r=1,2,\ldots$) let us define
\[
\left\Vert Z\left(\omega\right)\right\Vert _{\ns}:=\sup_{0\le t\le T}\left|Z_{t}\left(\omega\right)\right|,
\]
where $|\cdot |$ denotes the $l^2$ norm on $\R^r$.
Following \cite{LochPerkPro:2016}  we will identify two processes $X,Y:\Omega\times\left[0,T\right]\ra\R^{r}$
if
\[
\oP\left(\omega\in\Omega:\left\Vert X\left(\omega\right)-Y\left(\omega\right)\right\Vert _{\ns}>0\right)=0.
\]
This defines an equivalence relation, and we will write $\overline{L}_{0}\left(\R^{r}\right)$
(or $\overline{L}_{0}$ in short) for the space of its equivalence
classes.
We equip the space $\overline L_0(\R^r)$ with the distance
\begin{equation} \label{dinf_def}
  d_\infty(X,Y) := \overline E[\|X-Y\|_\infty \wedge 1],
\end{equation}
where $\overline E$ denotes an expectation operator defined for $Z\colon\Omega \rightarrow [0, \infty]$ by
\begin{equation*}
  \overline E[Z] := \inf \left\{\lambda > 0\,: \,\exists (H^n)_{n\in \N} \subseteq \mathcal{H}_{\lambda} \text{ s.t. } \forall \omega \in \Omega_\psi\,: \, \liminf_{n \to\infty} (\lambda + (H^n\cdot S)_T(\omega)) \ge Z(\omega) \right\}.
\end{equation*}
It can be shown that $(\overline L_0(\R^r), d_\infty)$ is a complete metric space and $(\overline{\mathcal{D}}(\R^r), d_\infty)$ is a closed subspace, where $\overline{\mathcal{D}}(\R^r)$ are those processes in $\overline L_0(\R^r)$ which have a c\`adl\`ag representative.\smallskip
\begin{definition} \label{conv_outer_set}
We will say that a sequence of $X_{n}\in\overline{L}_{0}$ \emph{converges
in the outer measure $\oP$ on a set $A\subset\Omega$ }to $X\in\overline{L}_{0}$
if for any $\varepsilon>0$
\[
\lim_{n\ra+\ns}\oP\left(\omega\in A:\left\Vert X_{n}\left(\omega\right)-X\left(\omega\right)\right\Vert _{\ns}>\varepsilon\right)=0.
\]
\end{definition}
For $q,M>0$ let us now define
\[
\Omega_{q,M}:=\left\{ \omega\in\Omega\,: |\left[S(\omega)\right]_{T} | \le q,\left\Vert S( \omega)\right\Vert _{\ns}\le M\right\} .
\]
It is easy to see that if for any $q,M>0,$ $X_{n}$
converges in outer measure $\oP$ on the set $\Omega_{q,M}$ to
$X\in\overline{L}_{0}$ then $X$ is unique.
For fixed $q,M>0$ let us now introduce the (pseudo-)distance $d_{\infty, q, M}$ on $\overline{L}_{0}$ which is given by
\begin{equation*}
d_{\infty, q,M}(X,Y):=\overline{E}[\Vert X-Y\Vert _{\infty}\wedge {\bf 1}_{\Omega_{q,M}}].
\end{equation*}
Now let us define for some fixed $\epsilon\in (0,1)$ the following metric on $\overline{L}_{0}$
\begin{equation}\label{eq:distance c\`{a}dl\`{a}g}
  d_{\infty,\psi}^{\epsilon}(X,Y):= \sum_{n,m=1}^\infty  2^{-(n/2+m)(1+\epsilon)} (\psi (2^m)\vee2^m)^{-1} d_{\infty,2^n,2^m}(X,Y)
\end{equation}
One of the main results of \cite{LochPerkPro:2016} is the following (see \cite[Theorem 4.2]{LochPerkPro:2016}).
\begin{theorem}\label{thm:integral}
  There exists two metric spaces $(\overline{H}_1, d_{\overline{H}_1})$ and $(\overline{H}_2, d_{\overline{H}_2})$ such that the (equivalence classes of) step functions (simple strategies) are dense in $\overline{H}_1$, $\overline{H}_2$ embedds into $\overline{\mathcal{D}}(\R^d)$ and the integral map $I \colon H \mapsto (H \cdot S) = \rbr{\int_{(0, t]} H_s \dd S_s}_{t \in [0, T]}$, defined for simple strategies in~\eqref{eq:step function integral}, has a continuous extension that maps $(\overline{H}_1, d_{\overline{H}_1})$ to $(\overline{H}_2, d_{\overline{H}_2})$. Moreover, one has the following continuity estimate
     \begin{equation}\label{eq:continuity c\`{a}dl\`{a}g}
   d_{\infty,\psi}^{\epsilon}((F\cdot S),(G\cdot S))
    \lesssim d_{\infty}(F,G)^{1/3},
\end{equation}
and one can take $d_{\overline{H}_1}=d_{\infty}$ and $d_{\overline{H}_2}=d_{\infty,\psi}^{\epsilon}$ (defined in~\eqref{eq:distance c\`{a}dl\`{a}g}). 
\end{theorem}
The metric space~$(\overline{H}_1, d_{\overline{H}_1})$ can be chosen to contain the left-continuous versions of adapted c\`adl\`ag processes (see \cite[Remark 4.3]{LochPerkPro:2016}). If we replace the filtration $\left({\mathcal F}_{t}\right)_{t\in\left[0,T\right]}$ by its right-continuous version, we can define $(\overline{H}_1, d_{\overline{H}_1})$ such that it contains at least the c\`agl\`ad adapted processes and, furthermore, such that if $(F_n)_{n \in \N} \subset \overline{H}_1$ is a sequence with $\sup_{\omega \in \Omega} \| F_n(\omega) - F(\omega) \|_{\infty} \to 0$ as $n \ra +\ns$, then $F\in \overline{H}_1$ and there exists a subsequence $(F_{n_k})$ with
    \begin{equation*}
      \lim_{k \to \infty} \| (F_{n_k} \cdot S)(\omega) - (F \cdot S)(\omega) \|_\infty = 0
    \end{equation*}
    for typical price paths $\omega\in \Omega.$
By \cite[Corollary 4.9]{LochPerkPro:2016} and (\ref{eq:continuity c\`{a}dl\`{a}g}) we also get
\begin{corollary}\label{cor:c\`{a}dl\`{a}g concentration estimate}
 For $a,q,m,M > 0$ and any $H\in{\overline{H}_1}$ one has
  \begin{align*}
    \oP\big(\big\{\|(H\cdot S)\|_\infty \geq a \big\} \cap\big\{\|H\|_\infty \leq m\big\}\cap \big\{ |[S]_T|\leq q \big\}
    &\cap\big \{ \Vert S \Vert_{\infty} \leq M \big\} \big) \\
    &\leq  (1+ 3 d M +2 d \psi (M))\frac{6\sqrt{q}+2+2M}{a}m.
  \end{align*}
\end{corollary}

\section{Quadratic variation of model-free c\`{a}dl\`{a}g price paths }

\subsection{Integration by parts formula for model-free price paths}

The existence of quadratic variation for typical $\omega\in\Omega$ in the sense outlined in the previous section is equivalent to the existence of quadratic
variation in the sense of Norvai\v{s}a \cite[Proposition~3]{Vovk_cadlag:2015}.
This also allows (see \cite[Sect. 7]{Vovk_cadlag:2015}) to apply
F\"{o}llmer's pathwise It\^{o} formula~\cite{Foellmer:1981} for typical
 price paths in $\Omega$ and in particular to define the pathwise
integral $\int_{(0, T]} f^{\prime}(\omega\rbr{s-})\dd \omega\rbr{s}$ (along the sequence of the Lebesgue partitions) for $f\in C^{2}\left(\R^{d},\R\right)$
when $d=1.$ To define the pathwise integral $\int_{(0, T]} \nabla f(\omega\rbr{s-})\dd \omega\rbr{s}$
when $d=2,3,\ldots$ and $f\in C^{2}\left(\R^{d},\R\right)$ or $f$
is even more general, path-dependent functional, one may use results
of Cont and Fourni\'{e} \cite{ContFournie:2010}.
We also have the following result which we will use to prove the next lemma.
\begin{lem} \label{fin_var_lemma}
Let $\omega,\tilde{\omega}:\left[0,T\right]\ra\R$ be two c\`{a}dl\`{a}g paths.
Assume that $\tilde{\omega}$ has finite total variation and let
us consider two integrals
\begin{enumerate}
\item the integral $\int_{(0,T]}\omega\left(t-\right)\dd{\tilde{\omega}\left(t\right)}$
understood as the Lebesgue-Stieltjes integral (with respect to the
measure $\dd{\tilde{\omega}}$ given by $\dd{\tilde{\omega}(a,b]}:=\tilde{\omega}\left(b\right)-\tilde{\omega}\left(a\right)$);
\item the integral $(F)\int_{(0,T]}\omega\left(t-\right)\dd{\tilde{\omega}\left(t\right)}$
understood as F\"{o}llmer's integral along the sequence of partitions $\left(\tau^{n}\right)_{n \in \N}$ such that for $n\in \N$, $\tau^n$ contains $n$th Lebesgue partition of the path $\omega$,
i.e.
\[
(F)\int_{(0,T]}\omega\left(t-\right)\dd{\tilde{\omega}\left(t\right)}:=\lim_{n\ra+\ns}\sum_{i=1}^{k_{n}}\omega\left(\tau_{i-1}^{n}\right)\left\{ \tilde{\omega}\left(\tau_{i}^{n}\right)-\tilde{\omega}\left(\tau_{i-1}^{n}\right)\right\} ,
\]
where
$
\tau^{n}:=\left\{ 0=\tau_{0}^{n}<\tau_{1}^{n}<\ldots<\tau_{k_{n}-1}^{n}<T = \tau_{k_{n}}^{n} < + \infty = \tau_{k_{n}+1}^{n} =  \tau_{k_{n}+2}^{n} = \ldots \right\}
$
and $\pi^n(\omega) \subset \tau^n$.
\end{enumerate}
These two integrals coincide, provided that the latter exists.
\end{lem}
\begin{proof}
Step 1. First, let us notice that for any $\varepsilon>0$ we may uniformly approximate $\omega$ by
a step function
\[
\omega^{\varepsilon}\left(t\right):=\sum_{i=1}^{N}\omega\left(t_{i-1}\right)\mathbf{1}_{[t_{i-1},t_{i})}\left(t\right)+\omega\left(T\right)\mathbf{1}_{\left\{ T\right\} }\left(t\right),
\]
where $0=t_{0}<t_{1}<\ldots<t_{N}=T,$ such that
\[
\left\Vert \omega^{\varepsilon}-\omega\right\Vert _{\ns}\le\varepsilon.
\]
To obtain such $\omega^{\varepsilon}$ we simply define $t_{0}:=0,$
$t_{i}:=\inf\left\{ t>t_{i-1}:\left|\omega\left(t\right)-\omega\left(t_{i-1}\right)\right|>\varepsilon\right\} $
for $i=1,2,\ldots$ such that $t_{i-1}<+\ns$ (we apply the convention
that $\inf\emptyset=+\ns$) and $N:=\max\left\{ i\in\left\{ 1,2,\ldots\right\} :t_{i-1}<T\right\} .$

Step 2. We have the estimate
\begin{eqnarray}
 &  & \left|\int_{(0,T]}\omega\left(t-\right)\dd{\tilde{\omega}\left(t\right)}-\int_{(0,T]}\omega^{\varepsilon}\left(t-\right)\dd{\tilde{\omega}\left(t\right)}\right|\nonumber \\
 &  & \le\int_{(0,T]}\left|\omega\left(t-\right) - \omega^{\varepsilon}\left(t-\right) \right|\left|\dd{\tilde{\omega}\left(t\right)}\right|\nonumber \\
 &  & \le\varepsilon\int_{(0,T]}\left|\dd{\tilde{\omega}\left(t\right)}\right|=\varepsilon\TTV{\tilde{\omega}}{\left[0,T\right]}{},\label{eq:one}
\end{eqnarray}
where $\TTV{\tilde{\omega}}{\left[0,T\right]}{}$ denotes the total variation of $\tilde{\omega}.$
Moreover
\[
\int_{(0,T]}\omega^{\varepsilon}\left(t-\right)\dd{\tilde{\omega}\left(t\right)}=\sum_{i=1}^{N}\omega^{\varepsilon}\left(t_{i-1}\right)\left(\tilde{\omega}\left(t_{i}\right)-\tilde{\omega}\left(t_{i-1}\right)\right).
\]
 We also have
\begin{eqnarray}
 &  & \left|\sum_{i=1}^{k_{n}}\omega\left(\tau_{i-1}^{n}\right)\left\{ \tilde{\omega}\left(\tau_{i}^{n}\right)-\tilde{\omega}\left(\tau_{i-1}^{n}\right)\right\} -\sum_{i=1}^{k_{n}}\omega^{\varepsilon}\left(\tau_{i-1}^{n}\right)\left\{ \tilde{\omega}\left(\tau_{i}^{n}\right)-\tilde{\omega}\left(\tau_{i-1}^{n}\right)\right\} \right|\nonumber \\
 &  & \le\sum_{i=1}^{k_{n}}\left|\omega\left(\tau_{i-1}^{n}\right)-\omega^{\varepsilon}\left(\tau_{i-1}^{n}\right)\right|\left|\tilde{\omega}\left(\tau_{i}^{n}\right)-\tilde{\omega}\left(\tau_{i-1}^{n}\right)\right|\nonumber \\
 &  & \le\varepsilon\TTV{\tilde{\omega}}{\left[0,T\right]}{}.\label{eq:two}
\end{eqnarray}

Step 3. Now let us consider the difference
\begin{eqnarray}
 &  & \int_{(0,T]}\omega^{\varepsilon}\left(t-\right)\dd{\tilde{\omega}\left(t\right)}-\sum_{i=1}^{k_{n}}\omega^{\varepsilon}\left(\tau_{i-1}^{n}\right)\left\{ \tilde{\omega}\left(\tau_{i}^{n}\right)-\tilde{\omega}\left(\tau_{i-1}^{n}\right)\right\} \nonumber \\
 &  & =\sum_{i=1}^{N}\omega^{\varepsilon}\left(t_{i-1}\right)\left(\tilde{\omega}\left(t_{i}\right)-\tilde{\omega}\left(t_{i-1}\right)\right)-\sum_{i=1}^{k_{n}}\omega^{\varepsilon}\left(\tau_{i-1}^{n}\right)\left\{ \tilde{\omega}\left(\tau_{i}^{n}\right)-\tilde{\omega}\left(\tau_{i-1}^{n}\right)\right\} .\label{eq:difference}
\end{eqnarray}

Let $\tau^{n}\left(t\right)$ denotes the first point $\tau_i^{n}$ in the partition
$\tau^{n}$ such that $\tau_i^{n}\ge t.$ From the definition
of the $n$th Lebesgue partition of $\omega$ it follows that for $t<T$
\[
\limsup_{n\ra+\ns}\tau^{n}\left(t\right) \le \inf\left\{ u>t:\omega\left(u\right)\neq\omega\left(t-\right)\right\} .
\]
By the definition of times $t_{1},t_{2},\ldots,t_{N-1}$ we have that
for any $t\in\left\{ t_{1},t_{2},\ldots,t_{N-1}\right\} ,$ $\omega\left(t\right)\neq\omega\left(t-\right)$
or $\omega\left(t\right)=\omega\left(t-\right)$ but $\omega$ is
not constant on any interval of the form $\left[t,u\right],$ $u\in\left(t,T\right]$
and $\lim_{n\ra+\ns}\pi^{n}\left(t\right)=t.$ Thus, for sufficiently
large $n$s
\[
t_{i}\le\tau^{n}\left(t_{i}\right)<t_{i+1} \text{ for } i=1,2,\ldots,N-1.
\]

Now, denoting by $k_{n}\left(t\right)$ such index that
$\tau^{n}\left(t\right)=\tau_{k_{n}\left(t\right)}^{n}$, for sufficiently
large $n$s and $i=2,\ldots,N$ we have $t_{i-2}\le\tau_{k_{n}\left(t_{i-1}\right)-1}^{n}<t_{i-1}$. Thus
$\omega^{\varepsilon}\left(\tau_{k_{n}\left(t_{i-1}\right)-1}^{n}\right)=\omega^{\varepsilon}\left(t_{i-2}\right)$ and since $\omega^{\varepsilon}$ is constant on $[t_{i-1},t_{i}),$
$i=1,2,\ldots,N,$
we obtain
\begin{eqnarray*}
 &  & \sum_{i=1}^{N}\omega^{\varepsilon}\left(t_{i-1}\right)\left(\tilde{\omega}\left(t_{i}\right)-\tilde{\omega}\left(t_{i-1}\right)\right)-\sum_{i=1}^{k_{n}}\omega^{\varepsilon}\left(\tau_{i-1}^{n}\right)\left\{ \tilde{\omega}\left(\tau_{i}^{n}\right)-\tilde{\omega}\left(\tau_{i-1}^{n}\right)\right\} \\
 &  & =\sum_{i=2}^{N}\left(\omega^{\varepsilon}\left(t_{i-1}\right)-\omega^{\varepsilon}\left(t_{i-2}\right)\right)\left(\tilde{\omega}\left(\tau^{n}\left(t_{i-1}\right)\right)-\tilde{\omega}\left(t_{i-1}\right)\right).
\end{eqnarray*}
Since $\lim_{n\ra+\ns}\tau^{n}\left(t\right)=t$ for $t\in\left\{ t_{1},t_{2},\ldots,t_{N-1}\right\} $
and $\tilde{\omega}$ is c\`adl\`ag, we finally get that the difference
(\ref{eq:difference}) tends to $0$ as $n\ra+\ns.$

From this and (\ref{eq:one}), (\ref{eq:two}) (taking $\varepsilon$ as close to $0$ as we wish) we get the assertion.

\end{proof}
Assume that $\omega \in D\rbr{[0,T],  \R^d}$ possesses quadratic variation along the sequence of its Lebesgue partitions, by which we mean that the sequence of discrete (co)variations defined in (\ref{discrete_qv}) converges for $i,j=1,2,\ldots, d$ in the uniform topology to $[S^{i},S^{j}](\omega)$. Applying (multidimensional) F\"{o}llmer's pathwise It\^{o} formula to
$f\left(x_{1},\ldots,x_{d}\right)=x_{i}x_{j},$ $i,j\in\left\{ 1,2,\ldots,d\right\} ,$ for $t \in [0, T]$
we obtain the integration by parts formula:
\begin{eqnarray}
\omega^{i}\rbr{t}\omega^{j}\rbr{t}-\omega^{i}\rbr{0}\omega^{j}\rbr{0}& = &
(F)\int_{\left(0,t\right]}\omega^{i}\rbr{s-}\dd\omega^{j}\rbr{s}+(F)\int_{\left(0,t\right]}\omega^{j}\rbr{s-}\dd\omega^{i}\rbr{s} \nonumber \\
&& +[S^{i},S^{j}]_t(\omega),\label{eq:int_by_parts_typical}
\end{eqnarray}
where $(F)\int_{\left(0,t\right]}\ldots$
denotes F\"{o}llmer's pathwise integral (along the sequence of the Lebesgue partitions).
Notice also that for typical price path $\omega$,
$(F)\int_{\left(0,t\right]}\omega^{i}\rbr{s-}\dd\omega^{j}\rbr{s}$
coincides with the model-free It\^o integral $\int_{\left(0,t\right]}S_{s-}^{i}\left(\omega\right)\dd{S_{s}^{j}}\left(\omega\right)$ introduced in the previous section, since it is the limit of the sums of the form
\[
\sum_{k=1}^{\infty}S_{\pi_{k-1}^{n}}^{i}(\omega)S_{\pi_{k-1}^{n}\wedge t,\pi_{k}^{n}\wedge t}^{j}(\omega)
\]
and $\sum_{k=1}^{\infty}S_{\pi_{k-1}^{n}}^{i}(\omega)\mathbf{1}_{(\pi_{k-1}^{n},\pi_{k}^{n}]}(t)$
converges uniformly to $S_{t-}^{i}(\omega)$ for $t\in\left(0,T\right],$ thus, by
Theorem \ref{thm:integral} the distance $d_{\infty, \psi}^{\epsilon}$  between F\"{o}llmer's integral and the model-free It\^o integral is $0$, and this implies that for typical price paths they coincide.
Thus (\ref{eq:int_by_parts_typical}) may be also viewed as the integration
by parts formula for the model-free It\^o integral.

Next, we have the following result, which we will need in the sequel.

\begin{lem} \label{int_by_parts_fin_var} Let $\tilde{S}:\Omega\times\left[0,T\right]\ra\R^d$
be such that the process $\tilde{S}$ is adapted (to the filtration $\left({\mathcal F}_{t}\right)_{t\in\left[0,T\right]}$), and for any  $\omega \in \Omega,$ $\tilde{\omega} := \tilde{S}\rbr{\omega}$ is
finite variation, c\`{a}dl\`{a}g path on $\left[0,T\right].$  Then for typical
path $\omega\in\Omega$ and $i,j=1,2,\ldots,d$, $t \in [0, T]$, the following integration
by parts formula holds
\begin{align}
\omega^i\left(t\right)\tilde{\omega}^j\left(t\right)-\omega^i\left(0\right)\tilde{\omega}^j\left(0\right) & =\int_{\left(0,t\right]}\tilde{S}_{s-}^j\rbr{\omega}\dd{S_{s}^i}\left(\omega\right)+\int_{\left(0,t\right]}\omega^i\left(s-\right)\dd{\tilde{\omega}^j\left(s\right)} \nonumber \\ &\quad +\sum_{0<s\le t}\Delta\omega^i\left(s\right)\Delta\tilde{\omega}^j\left(s\right),\label{eq:int_by_parts_fin_typ}
\end{align}
where $\int_{\left(0,t\right]}\tilde{S}_{s-}^j\rbr{\omega}\dd{S_{s}^i}\left(\omega\right)$
denotes the model-free It\^o integral
and $\int_{\left(0,t\right]}\omega^i\left(s-\right)\dd{\tilde{\omega}^j\left(s\right)} $
denotes the Lebesgue-Stieltjes integral which coincides with the F\"{o}llmer integral.
\end{lem}
\begin{proof} We will use F\"{o}llmer's pathwise integration by parts formula (\ref{eq:int_by_parts_typical}).
Since (as we already noticed) the model-free It\^o integral $\int_{\left(0,t\right]}\tilde{S}_{s-}^j\rbr{\omega}\dd{S_{s}^i}\left(\omega\right)$ coincides for typical price paths with the F\"{o}llmer integral (along the sequence of the Lebesgue partitions $\rbr{\pi^n}_{n \in \N}$ of $\rbr{\omega, \tilde{\omega}} \in \R ^{2d}$) and since (by Lemma \ref{fin_var_lemma}) F\"{o}llmer's integral
$(F)\int_{\left(0,t\right]}\omega^i\left(s-\right)\dd{\tilde{\omega}^j\left(s\right)} $ (along the same sequence of partitions)
coincides with the classical
Lebesgue-Stieltjes integral $\int_{\left(0,t\right]}\omega^i\left(s-\right)\dd{\tilde{\omega}^j\left(s\right)} ,$ to obtain the thesis we need only to prove that for $i,j =1,2,\ldots, d$, $t \in [0, T]$,
$$\sum_{k=1}^{\infty}S_{\pi_{k-1}^{n}\wedge t,\pi_{k}^{n}\wedge t}^i(\omega){S}_{\pi_{k-1}^{n}\wedge t,\pi_{k}^{n}\wedge t}^j(\tilde{\omega}) = \sum_{k=1}^{\infty}S_{\pi_{k-1}^{n}\wedge t,\pi_{k}^{n}\wedge t}^i(\omega)\tilde{S}_{\pi_{k-1}^{n}\wedge t,\pi_{k}^{n}\wedge t}^j({\omega})$$
converges to $\sum_{0<s\le t}\Delta\omega^i\left(s\right)\Delta\tilde{\omega}^j\left(s\right)$ (uniformly in $t$).
The proof of this fact follows from the properties of Lebesgue partitions
and the Schwartz inequality. Let $\varepsilon>0$ be such that there is no jump of $\tilde{\omega}$ of size equal $\varepsilon$ and let $I^{\varepsilon,n},$
$n\in\N$, be the sequence of all indices $k\in\N$ for which $\left|S_{\pi_{k-1}^{n}\wedge t,\pi_{k}^{n}\wedge t}^j(\tilde{\omega})\right|>\varepsilon.$
We have
\begin{align*}
 & \left|\sum_{k=1}^{\infty}S_{\pi_{k-1}^{n}\wedge t,\pi_{k}^{n}\wedge t}^i(\omega){S}_{\pi_{k-1}^{n}\wedge t,\pi_{k}^{n}\wedge t}^j(\tilde{\omega})-\sum_{k\in I^{\varepsilon,n}}S_{\pi_{k-1}^{n}\wedge t,\pi_{k}^{n}\wedge t}^i(\omega){S}_{\pi_{k-1}^{n}\wedge t,\pi_{k}^{n}\wedge t}^j(\tilde{\omega})\right|\\
 & =\left|\sum_{k\in\N\setminus I^{\varepsilon,n}}S_{\pi_{k-1}^{n}\wedge t,\pi_{k}^{n}\wedge t}^i(\omega){S}_{\pi_{k-1}^{n}\wedge t,\pi_{k}^{n}\wedge T}^j(\tilde{\omega})\right|\\
 & \le \sqrt{\sum_{k\in\N\setminus I^{\varepsilon,n}}S_{\pi_{k-1}^{n}\wedge t,\pi_{k}^{n}\wedge t}^i(\omega)^{2}} \sqrt{\sum_{k\in\N\setminus I^{\varepsilon,n}}S_{\pi_{k-1}^{n}\wedge t,\pi_{k}^{n}\wedge t}^j(\tilde{\omega})^{2}}\\
 & \le\sqrt{\varepsilon}\sqrt{\sum_{k=1}^{\infty}\left|S_{\pi_{k-1}^{n}\wedge t,\pi_{k}^{n}\wedge t}^j(\tilde{\omega})\right|}\sqrt{\sum_{k=1}^{\infty}S_{\pi_{k-1}^{n}\wedge t,\pi_{k}^{n}\wedge t}^i(\omega)^{2}}.
\end{align*}
Notice that $\sum_{k=1}^{\infty}\left|S_{\pi_{k-1}^{n}\wedge t,\pi_{k}^{n}\wedge t}^j(\tilde{\omega})\right|$
is bounded by the total variation of $\tilde{\omega}$ while $$\rbr{\sum_{k=1}^{\infty}S_{\pi_{k-1}^{n}\wedge t,\pi_{k}^{n}\wedge t}^i(\omega)^{2}}_{t \in [0, T]}$$
converges to the quadratic variation of ${\omega^i}$ as $n\ra+\ns$ (the convergence still holds though partitions $\rbr{\pi^n}_{n \in \N}$ may be finer than the Lebesgue partitions of $\omega$, see the proofs of \cite[Corollary 3.11]{LochPerkPro:2016},  \cite[Theorem 2]{Vovk_cadlag:2015}).
Finally notice that
\[
\lim_{n\ra\ns}\sum_{k\in I^{\varepsilon,n}}S_{\pi_{k-1}^{n}\wedge t,\pi_{k}^{n}\wedge t}^i(\omega){S}_{\pi_{k-1}^{n}\wedge t,\pi_{k}^{n}\wedge t}^j(\tilde{\omega}) =\sum_{0<s\le t,\left|\Delta\tilde{\omega}\left(s\right)\right|>\varepsilon}\Delta\omega^i\left(s\right)\Delta\tilde{\omega}^j\left(s\right).
\]
(recall that there is no jump of $\tilde{\omega}$ of size equal $\varepsilon$). Since there is only finite number of jumps of $\tilde{\omega}$ of size greater than $\varepsilon$ and $\varepsilon>0$ may be as close to $0$ as we wish (since $\tilde{\omega}$ has only countable number of jumps), the result
follows.
\end{proof}

\begin{rem} It is straightforward to generalise (\ref{eq:int_by_parts_typical})
and (\ref{eq:int_by_parts_fin_typ}) to the case when the integrand
and the integrator are sums of a typical and a finite-variation path.
\end{rem}

\subsection{Quadratic variation along the partitions obtained from stopping times}

Let $\tau=\left\{ 0=\tau_{0}\le \tau_{1}\le\ldots\right\} $ be a (possibly
random) partition of the interval $\left[0,T\right],$ by which we
will mean that $\tau_i,$ $i=0,1,\ldots,$ is infinite, non-decreasing sequence of elements of $\left[0,T\right]\cup\left\{ +\ns\right\} $
such that from some $i\in\left\{ 1,2,\ldots\right\} $ on $\tau_{i}=\tau_{i+1}=\ldots=+\infty.$ To avoid redundancies, we may also  assume (as in \cite{Obloj_local:2015}, though this will not change the reasoning) that for all $i = 0,1,\ldots,$ $\tau_i < \tau_{i+1}$ whenever $\tau_i <+\infty.$
Similarly as Davis, Ob\l\'{o}j and Siorpaes in \cite{Obloj_local:2015},
for any $\omega\in\Omega$ we will denote
\[
O_{T}\left(\omega,\tau\right):=\sup\left\{ \left|\omega\left(t\right)-\omega\left(s\right)\right|\,:\, s,t\in\left[\tau_{i-1},\tau_{i}\right)\cap\left[0,T\right],i\in\left\{ 1,2,\ldots\right\} \right\} .
\]
Moreover, similarly to (\ref{discrete_qv}), for $\omega \in \Omega$, $i,j =1,2,\ldots,d$ and $t \in [0,T]$ let us define
\begin{equation} \label{discrete_qv1}
Q_{t}^{i,j,\tau}(\omega):=\sum_{k=1}^{\infty}S_{\tau_{k-1}\wedge t, \tau_{k}\wedge t}^{i}(\omega)S_{\tau_{k-1}\wedge t, \tau_{k}\wedge t}^{j}(\omega).
\end{equation}
\begin{definition}
We will say that the partition $\tau=\left\{ 0=\tau_{0}\le\tau_{1}\le\ldots\right\} $
is an \emph{optional} one (with respect to a given filtration) if all $\tau_{1}\le\tau_{2}\le\ldots$ are
stopping times (with respect to this filtration).
\end{definition}
Let us recall also the definition of the convergence in outer measure on a given set (Definition \ref{conv_outer_set}). Now we are ready to state
\begin{theorem} Let $\tau^{n},$ $n=1,2\ldots,$ be a sequence of optional
partitions (with respect to the filtration $\left({\mathcal F}_{t}\right)_{t\in\left[0,T\right]}$) of the interval $\left[0,T\right]$ such that for all $q,M>0$,
$O_{T}\left(\cdot,\tau^{n}\right)$ converges in outer measure to $0$
(as $n\ra0$) on $\Omega_{q,M}.$ Then for any $q,M>0$, $i,j=1,2,\ldots,d$, the sequence
$\rbr{Q^{i,j,\tau^n}}_{n\in \N}$ converges in outer
measure $\oP$ on $\Omega_{q,M}$ to $\sbr{S^i, S^j}$, i.e. the quadratic (co)variation obtained along the sequence of the Lebesgue partitions.
\end{theorem}
\begin{proof} First we will prove the thesis for $i=j\in \cbr{1,2,\ldots,d}$. Let $\tau^{n}=\left\{ 0=\tau_{0}^{n}\le\tau_{1}^{n}\le\ldots\right\} $
and take
\[
F^{n}_t\left(\omega\right):=\sum_{i=1}^{+\ns}S^i_{\tau_{i-1}^{n}}\rbr{\omega}\mathbf{1}_{\left[\tau_{i-1}^{n},\tau_{i}^{n}\right)}\left(t\right).
\]
By the definition of $O_{T}\left(\omega,\tau^{n}\right)$ we have
\[
\left\Vert F^{n}(\omega)-S^i(\omega)\right\Vert _{\ns}\le O_{T}\left(\omega,\tau^{n}\right).
\]
Now, using the integration by parts formula (\ref{eq:int_by_parts_typical}) for any $t\in\left[0,T\right]$ we calculate
\begin{equation}
Q^{i,i,\tau^n}_t(\omega)-\left[S^i, S^i\right]_{t}(\omega)=2\int_{\left(0,t\right]}\left(S_{s-}^i(\omega)-F_{s-}^{n}(\omega)\right)\mathrm{d}S_s^i(\omega),\label{eq:q_var_int}
\end{equation}
where $\int_{\left(0,t\right]}\left(S_{s-}^i(\omega)-F_{s-}^{n}(\omega)\right)\mathrm{d}S_s^i(\omega)$ denotes the model-free It\^o integral. (Note that this integral is well defined since the partition $\tau^{n}$ is optional one.) Now,
fixing $\varepsilon,\delta>0,$ and applying Corollary \ref{cor:c\`{a}dl\`{a}g concentration estimate} we get for some constant $C_{q,M}$
depending on $q$ and $M$ (and $\psi$) only
\begin{align*}
 & \oP\left(\omega\in\Omega_{q,M}\,:\,\left\Vert Q^{i,i,\tau^n}_t(\omega)-\left[S^i, S^i\right]_{t}(\omega)\right\Vert _{\ns}>\varepsilon\right)\\
 & =\oP\left(\omega\in\Omega_{q,M}\,:\,\left\Vert \rbr{\int_{\left(0,t\right]}\left(S_{s-}^i(\omega)-F_{s-}^{n}(\omega)\right)\mathrm{d}S_s^i(\omega)}_{t \in [0, T]}\right\Vert _{\ns}>\varepsilon/2\right)\\
 & \le\oP\left(\omega\in\Omega_{q,M}\,:\, O_{T}\left(\omega,\tau^{n}\right)>\delta\right)\\
 & \quad+\oP\left(\omega\in\Omega_{q,M}\,:\, O_{T}\left(\omega,\tau^{n}\right)\le\delta,\left\Vert \rbr{\int_{\left(0,t\right]}\left(S_{s-}^i(\omega)-F_{s-}^{n}(\omega)\right)\mathrm{d}S_s^i(\omega)}_{t \in [0, T]}\right\Vert _{\ns}>\varepsilon/2\right)\\
 & \le\oP\left(\omega\in\Omega_{q,M}\,:\, O_{T}\left(\omega,\tau^{n}\right)>\delta\right)+C_{q,M}\frac{\delta}{\varepsilon}.
\end{align*}
Since $\oP\left(\omega\in\Omega_{q,M}:O_{T}\left(\omega,\tau^{n}\right)>\delta\right)\ra0$
as $n\ra+\ns$ and $\delta/\varepsilon$ may be chosen arbitrary close
to $0,$ we get the convergence result.

To obtain the convergence for  $i,j\in \cbr{1,2,\ldots,d}$, $i \neq j$, we will use polarization. More precisely, we define
\[
H^{n}_t\left(\omega\right):=\sum_{i=1}^{+\ns}\rbr{S^i_{\tau_{i-1}^{n}}\rbr{\omega} + S^j_{\tau_{i-1}^{n}}\rbr{\omega}}\mathbf{1}_{\left[\tau_{i-1}^{n},\tau_{i}^{n}\right)}\left(t\right),
\]
and using integration by parts we obtain that the process
\[
R^{i,j,\tau^n}_t (\omega ) :=\sum_{k=1}^{\infty}\rbr{S_{\tau_{k-1}^{n}\wedge t, \tau_{k}^{n}\wedge t}^{i}(\omega)+S_{\tau_{k-1}^{n}\wedge t, \tau_{k}^{n}\wedge t}^{j}(\omega)}^2
\]
converges in outer
measure $\oP$ on $\Omega_{q,M}$ to $\sbr{S^i +S^j, S^i +S^j} = \sbr{S^i, S^i} + 2 \sbr{S^i, S^j} + \sbr{S^j,S^j}$ since
\[
R^{i,j,\tau^n}_t(\omega)-\left[S^i+ S^j, S^i + S^j\right]_{t}(\omega)=2\int_{\left(0,t\right]}\left(S_{s-}^i(\omega) + S_{s-}^j(\omega)-H_{s-}^{n}(\omega) \right)\mathrm{d}\rbr{S_s^i(\omega)+S_s^j(\omega)}.
\]
Similarly, one proves that the process  $$T^{i,j,\tau^n}_t(\omega) = \sum_{k=1}^{\infty}\rbr{S_{\tau_{k-1}^{n}\wedge t, \tau_{k}^{n}\wedge t}^{i}(\omega)-S_{\tau_{k-1}^{n}\wedge t, \tau_{k}^{n}\wedge t}^{j}(\omega)}^2$$
converges in outer
measure $\oP$ on $\Omega_{q,M}$ to $\sbr{S^i -S^j, S^i -S^j} = \sbr{S^i, S^i} - 2 \sbr{S^i, S^j} + \sbr{S^j,S^j}$. Thus we obtain that the difference of the processes $R^{i,j,\tau^n}$ and $T^{i,j,\tau^n}$ which is equal $4Q^{i,j,\tau^n}$ converges in outer
measure $\oP$ on $\Omega_{q,M}$ to $4\sbr{S^i ,S^j} $.
\end{proof}

When we assume some stronger mode of convergence of $O_{T}\left(\cdot,\tau^{n}\right)$
then, naturally, we may expect to obtain a stronger mode of convergence
of $Q^{i,j,\tau^n}.$ Let us recall the definition
of distances $d_{\infty}$ and $d_{\infty,\psi}^{\epsilon}$ (defined in (\ref{dinf_def}) and (\ref{eq:distance c\`{a}dl\`{a}g})). We have the following result.
\begin{theorem} Let $\tau^{n},$ $n=1,2\ldots,$ be a sequence of optional
partitions of the interval $\left[0,T\right]$, $i\in \cbr{1,2,\ldots,d}$ and
\[
F^{n}_t\left(\omega\right):=\sum_{i=1}^{+\ns}S^i_{\tau_{i-1}^{n}}\rbr{\omega}\mathbf{1}_{\left[\tau_{i-1}^{n},\tau_{i}^{n}\right)}\left(t\right).
\]
Assume that  $\lim_{n\ra+\ns}d_{\ns}\left(F^{n},S^i\right)=0$ then for  any $\epsilon \in (0,1),$
\begin{equation} \label{qvar_convii}
d_{\infty,\psi}^{\epsilon}\rbr{Q^{i,i,\tau^n},\sbr{S^i, S^i}}\ra0\quad\mbox{as }n\ra0.
\end{equation}

Similarly, for $i,j \in \cbr{1,2,\ldots,d}$ define
\[
H^{n}_t\left(\omega\right):=\sum_{i=1}^{+\ns}\rbr{S^i_{\tau_{i-1}^{n}}\rbr{\omega} + S^j_{\tau_{i-1}^{n}}\rbr{\omega}}\mathbf{1}_{\left[\tau_{i-1}^{n},\tau_{i}^{n}\right)}\left(t\right)
\]
and
\[
J^{n}_t\left(\omega\right):=\sum_{i=1}^{+\ns}\rbr{S^i_{\tau_{i-1}^{n}}\rbr{\omega} - S^j_{\tau_{i-1}^{n}}\rbr{\omega}}\mathbf{1}_{\left[\tau_{i-1}^{n},\tau_{i}^{n}\right)}\left(t\right).
\]
Assume that  $\lim_{n\ra+\ns}d_{\ns}\left(H^{n},S^i + S^j\right)=0$ and $\lim_{n\ra+\ns}d_{\ns}\left(J^n,S^i - S^j\right)=0$ then for any $\epsilon \in (0,1),$
\begin{equation} \label{qvar_convij}
d_{\infty,\psi}^{\epsilon}\rbr{Q^{i,j,\tau^n},\sbr{S^i, S^j}}\ra0\quad\mbox{as }n\ra0.
\end{equation}

\end{theorem}
\begin{proof} Convergence (\ref{qvar_convii}) follows immediately from (\ref{eq:q_var_int})
and Theorem \ref{thm:integral} applied to $F_t=\rbr{0,0, \ldots, F_{t-}^n, \ldots, 0}$ ($i$th component of $F_t$ is equal $F_{t-}^n$ and all other components are equal $0$), and $G_t = \rbr{0,0, \ldots, S_{t-}^i, \ldots, 0}$ ($i$th component of $G_t$ is equal $S_{t-}^i$ and all other components are equal $0$). Convergence (\ref{qvar_convij}) follows in a similar way from polarization, integration by parts and Theorem \ref{thm:integral}.
\end{proof}

\subsection{Quadratic variation expressed in terms of the truncated variation}

The theorem which we will prove in this subsection provides one more formula
for the quadratic variation of model-free c\`{a}dl\`{a}g price paths. This will be a model-free counterpart of the formula obtained for c\`adl\`ag semimartingles in \cite{LochQVTV:2017}. More
precisely, we will obtain partition-independent formula for the continuous
part of the quadratic (co)variation along the sequence
of Lebesgue partitions, which we will denote by $\left\langle S^i, S^j \right\rangle.$
It is formally defined for $i,j \in \cbr{1,2,\ldots, d}$ as
\[
\left\langle S^i, S^j  \right\rangle _{t}(\omega)=\left[S^i, S^j \right]_{t}(\omega)-\sum_{0<s\le t} \Delta\omega^i(s) \Delta \omega^j (s).
\]

To state our result we need to introduce the notion of the truncated variation (with the truncation parameter $c\ge0$) of
a c\`{a}dl\`{a}g path $f:\left[0,T\right]\ra\R$ which is for $t \in [0,T]$ defined as
\[
\TTV{f}{\left[0,t\right]}c:=\sup_{n}\sup_{0\le t_{0}<t_{1}<\ldots<t_{n}\le t}\sum_{i=1}^{n}\max\left\{ \left|f\left(t_{i}\right)-f\left(t_{i-1}\right)\right|-c,0\right\} .
\]
Notice that $\TTV{f}{\left[0,t\right]}c$ does not depend on any partition, since it is the supremum over \emph{all} partitions of the interval $[0, t].$

\begin{theorem}\label{q_var_TV_representation} Let us fix $q,M>0$. For $i,j \in \cbr{1,2,\ldots, d}$ the following convergences hold:
\begin{equation} \label{pol1}
c\cdot \TTVemph{S^i(\omega)}{\left[0,\cdot \right]}c \ra_{c\ra0+}^{\oP} \left\langle S^i, S^i \right\rangle (\omega),
\end{equation}
and
\begin{equation} \label{pol2}
c\rbr{\TTVemph{S^i(\omega)+S^j(\omega)}{\left[0,\cdot \right]}c - \TTVemph{S^i(\omega)-S^j(\omega)}{\left[0,\cdot \right]}c} \ra_{c\ra0+}^{\oP} 4 \left\langle S^i, S^j \right\rangle  (\omega),
\end{equation}
where $\left\langle S^i, S^j \right\rangle $ denotes the
continuous part the quadratic (co)variation (defined along the sequence
of the Lebesgue partitions) and $\ra_{c\ra0+}^{\oP}$ denotes the convergence
in the outer measure $\oP$ as $c\ra0+$ on the set $\Omega_{q,M}.$
\end{theorem}
\begin{proof} Using construction in \cite[Sect. 2]{Lochowski_stoch_integral:2014}
we know that for any $\omega \in D\rbr{[0, T], \R^d}$ and
$c>0$ there exists $\omega^{c}\in D\rbr{[0, T], \R^d}$
such that
\begin{enumerate}
\item $\omega^{c}$ has finite total variation;
\item $\omega^{c}\left(0\right)=\omega\left(0\right);$
\item for every $t\in\left[0,T\right],$ $i \in \cbr{1,2,\ldots, d}$, $\left|\omega^i\left(t\right)-\rbr{\omega^{c}}^i\left(t\right)\right|\le c;$
\item for every $t\in\left(0,T\right],$ $i \in \cbr{1,2,\ldots, d}$, $\left|\Delta\rbr{\omega^{c}}^i\left(t\right)\right|\le\left|\Delta\omega^i\left(t\right)\right|,$
\item the process $S^c(\omega) := \omega^c$ is adapted to the filtration $\left({\mathcal F}_{t}\right)_{t\in\left[0,T\right]}.$
\end{enumerate}
Moreover (see \cite[Lemma 5.1]{Lochowski_stoch_integral:2014} for $i \in \cbr{1,2,\ldots, d}$, $t \in [0, T]$ we have
\begin{equation}
\TTV{\omega^i}{\left[0,t\right]}{2c}\le\TTV{\rbr{\omega^{c}}^i}{\left[0,t\right]}{}\le\TTV{\omega^i}{\left[0,t\right]}{2c}+2c\label{eq:zero}
\end{equation}
and
\begin{equation}
c\cdot\TTV{\rbr{\omega^{c}}^i}{\left[0,t\right]}{}=\int_{\left(0,t\right]}\left(\omega^i-\rbr{\omega^{c}}^i\right)\dd{\rbr{\omega^{c}}^i},\label{eq:first}
\end{equation}
where $\int_{\left(0,t\right]}\left(\omega^i-\rbr{\omega^{c}}^i\right)\dd{\rbr{\omega^{c}}^i}$
denotes the standard Lebesgue-Stieltjes integral and $\TTV{\rbr{\omega^{c}}^i}{\left[0,t\right]}{}$
denotes the total variation of $\rbr{\omega^{c}}^i$ (recall that $\omega^{c}$
has finite total variation ).

Similarly as in \cite[proof of Theorem 1]{LochQVTV:2017} we calculate
\begin{eqnarray}
\int_{\left(0,t\right]}\left(\omega^i-\rbr{\omega^{c}}^i\right)\dd{\rbr{\omega^{c}}^i} & = & \int_{\left(0,t\right]}\left(\omega^i\left(s-\right)-\rbr{\omega^{c}}^i\left(s-\right)+\Delta\left(\omega^i\left(s\right)-\rbr{\omega^{c}}^i\left(s\right)\right)\right)\dd{\rbr{\omega^{c}}^i\left(s\right)}\nonumber \\
 & = & \int_{\left(0,t\right]}\omega^i\left(s-\right)\dd{\rbr{\omega^{c}}^i\left(s\right)}-\int_{\left(0,t\right]}\rbr{\omega^{c}}^i\left(s-\right)\dd{\rbr{\omega^{c}}^i\left(s\right)}\nonumber \\
 &  & +\sum_{0<s\le t}\Delta\left(\omega^i\left(s\right)-\rbr{\omega^{c}}^i\left(s\right)\right)\Delta\omega^i\left(s\right).\label{eq:second}
\end{eqnarray}
By Lemma \ref{int_by_parts_fin_var}
\begin{eqnarray}
\int_{\left(0,t\right]}\omega^i\left(s-\right)\dd{\rbr{\omega^c}^i\left(s\right)} & = & \rbr{\omega^{c}}^i \left(t\right)\omega^i\left(t\right)-\rbr{\omega^c}^i\left(0\right)\omega^i\left(0\right)-\int_{\left(0,t\right]}\rbr{S^c}^i_{s-}\rbr{\omega}\dd S_s^i\rbr{\omega}\nonumber \\
 &  & -\sum_{0<s\le t}\Delta\rbr{\omega^c}^i\left(s\right)\Delta\omega^i\left(s\right),\label{eq:change_of_var}
\end{eqnarray}
where $\int_{\left(0,t\right]}\rbr{S^c}^i_{s-}\rbr{\omega}\dd S_s^i\rbr{\omega}$ denotes the model-free It\^{o} integral.

By $\ra_{c\ra0+}$ we will denote the uniform convergence in $t\in [0, T]$ for \emph{typical}
$\omega\in\Omega$ as $c\ra0+.$

Using $\left|\Delta\rbr{\omega^c}^i\left(t\right)\right|\le\left|\Delta\omega^i\left(t\right)\right|$
and $\sum_{0<s\le t}\left(\Delta\omega^i\left(s\right)\right)^{2}\le\left[S^i, S^i\right]_{t}(\omega),$
by the dominated convergence we get
\[
\sum_{0<s\le t}\Delta\rbr{\omega^c}^i\left(s\right)\Delta\omega^i\left(s\right)\ra_{c\ra0+}\sum_{0<s\le t}\left(\Delta\omega^i\left(s\right)\right)^{2}.
\]
Now, by (\ref{eq:change_of_var}), the fact that $\Vert S^i-\rbr{S^c}^i\Vert_{\ns}\le c$
and Corollary \ref{cor:c\`{a}dl\`{a}g concentration estimate} we get the convergence
\begin{eqnarray*}
\int_{\left(0,t\right]}\omega^i\left(s-\right)\dd{\rbr{\omega^c}^i\left(s\right)} & \ra_{c\ra0+}^{\oP} & \left(\omega^i\left(t\right)\right)^{2}-\left(\omega^i\left(0\right)\right)^{2}-\int_{\left(0,t\right]}S^i_{s-}\rbr{\omega}\dd S^i_s\rbr{\omega} \\&  &-\sum_{0<s\le t}\left(\Delta\omega^i\left(s\right)\right)^{2}.
\end{eqnarray*}

Recall also that $\omega^{c}$ has finite total variation,
thus the Lebesgue-Stieltjes integral rules apply and we have
\begin{eqnarray*}
\int_{\left(0,t\right]}\rbr{\omega^c}^i\left(s-\right)\dd{\rbr{\omega^c}^i\left(s\right)} & = & \frac{1}{2}\left(\left(\rbr{\omega^c}^i\left(t\right)\right)^{2}-\left(\rbr{\omega^c}^i\left(0\right)\right)^{2}\right)-\frac{1}{2}\sum_{0<s\le t}\left(\Delta\rbr{\omega^c}^i\left(s\right)\right)^{2}\\
 &  & \ra_{c\ra0+}\frac{1}{2}\left(\left(\omega^i\left(t\right)\right)^{2}-\left(\omega^i\left(0\right)\right)^{2}\right)-\frac{1}{2}\sum_{0<s\le t}\left(\Delta\omega^i\left(s\right)\right)^{2}.
\end{eqnarray*}
We also have
\[
\sum_{0<s\le t}\Delta\left(\omega^i\left(t\right)-\rbr{\omega^c}^i\left(s\right)\right)\Delta\omega^i\left(s\right)\ra_{c\ra0+}0.
\]
Thus, from (\ref{eq:second}) and last three convergences we get
\begin{eqnarray}
\int_{\left(0,t\right]}\left(\omega^i-\rbr{\omega^c}^i\right)\dd{\rbr{\omega^c}^i} & \ra_{c\ra0+}^{\oP} & \left(\omega^i\left(t\right)\right)^{2}-\left(\omega^i\left(0\right)\right)^{2} \nonumber \\
& &-\int_{\left(0,t\right]}S^i_{s-}\rbr{\omega}\dd S_s^i\rbr{\omega}-\sum_{0<s\le t}\left(\Delta\omega^i\left(s\right)\right)^{2}\nonumber \\
 &  & -\frac{1}{2}\left(\left(\omega^i\left(t\right)\right)^{2}-\left(\omega^i\left(0\right)\right)^{2}\right)+\frac{1}{2}\sum_{0<s\le t}\left(\Delta\omega^i\left(s\right)\right)^{2}.\label{eq:uff}
\end{eqnarray}
By the integration by parts
\begin{equation}
\int_{\left(0,t\right]}S^i_{s-}\rbr{\omega}\dd S^i_s\rbr{\omega}=\frac{1}{2}\left(\left(\omega^i\left(t\right)\right)^{2}-\left(\omega^i\left(0\right)\right)^{2}\right)-\frac{1}{2}\left[S^i, S^i\right]_{t}.\label{eq:ufff}
\end{equation}
Finally, from (\ref{eq:uff}) and (\ref{eq:ufff}), recalling that
$\left[S^i, S^i \right]_{t}=\left\langle S^i, S^i \right\rangle _{t}+\sum_{0<s\le t}\left(\Delta\omega^i\left(s\right)\right)^{2}$
we have
\[
\int_{(0,t]}\left(\omega^i-\rbr{\omega^c}^i\right)\dd{\rbr{\omega^c}^i}\ra_{c\ra0+}^{\oP}\frac{1}{2}\left\langle S^i, S^i\right\rangle _{t}
\]
and from (\ref{eq:zero}) and (\ref{eq:first}) we get (\ref{pol1}):
\[
2c\cdot\TTV{S^i(\omega)}{\left[0,t\right]}{2c}\ra_{c\ra0+}^{\oP}\left\langle S^i, S^i \right\rangle _{t}(\omega).
\]

To obtain convergence (\ref{pol2}) we apply polarization.
\end{proof}

\begin{rem}
From Theorem \ref{q_var_TV_representation} an analogous result follows for any local
martingale, since, for any probability measure $\P$ such that the
coordinate process $S$ is a local martingale and any $B\in{\mathcal F}_{T}$
we have $\P\left(B\right)\le\oP\left(B\right),$ see \cite[Proposition 2.4]{LochPerkPro:2016}.
This, with the help of simple inequality
\[
c\left|\TTVemph{S^i+A}{\left[0,t\right]}c-\TTVemph {S^i}{\left[0,t\right]}c\right|\le c\cdot \TTVemph A{\left[0,t\right]}{},
\]
where $A$ is a real process with finite total variation (see \cite[Fact 17 and ineq. (2.14)]{LochowskiMilosSPA:2013})
proves similar result for any c\`adl\`ag semimartingale. However, for c\`adl\`ag  semimartingales
a stronger (almost sure) convergence may be obtained --- see \cite[Theorem 1]{LochQVTV:2017}.
\end{rem}

\subsection{Quadratic variation vs truncated variation of deterministic paths }

Let now $\omega:\left[0,T\right]\ra\R$ be a c\`{a}dl\`{a}g deterministic
path. Theorem \ref{q_var_TV_representation} raises the question what
is the relation between the existence of the limit $c\cdot\TTV{\omega}{\left[0,T\right]}c$
as $c\ra0+$ and the existence of the quadratic variation along some
sequence of partitions of $\omega.$ It is well known that there exist
continuous paths for which one may obtain \emph{arbitrary} non-decreasing
(and starting from $0$) quadratic variations, as limits of discrete
quadratic variations along appropriately chosen (refining) sequence
of partitions (cf. \cite[Theorem 7.1]{Obloj_local:2015}).

In fact, trajectories of a standard (one-dimensional) Brownian motion $B$ provide (with
probability $1$) examples of such paths.
This fact, together with the fact that for the standard Brownian motion
$B$ one has almost surely (see \cite[Theorem 1]{LochowskiMilosSPA:2013})
\[
c\cdot\TTV B{\left[0,T\right]}c\ra_{c\ra0+}T,
\]
proves that the existence of the limit $c\cdot\TTV{\omega}{\left[0,T\right]}c$
\emph{does not imply} the existence of the same quadratic variation
along a given refining sequence of partitions of $\omega.$

Also, for any continuous
function $\omega:\left[0,T\right]\ra\R,$ using the Darboux property
it is easy to construct a refining sequence of partitions such that
the sequence of corresponding discrete quadratic variations tends
to $0.$ Thus, the
existence of the quadratic variation along a given refining sequence of partitions
of $\omega$ \emph{does not imply} the existence of the finite limit
$c\cdot\TTV{\omega}{\left[0,T\right]}c$ as $c\ra0+$ (this may for
example tend to $+\ns$ for very irregular $\omega,$ for example
for $\omega$ which has $n^{2}$ oscillations of size $1/\sqrt{n}$
on the interval $\left[1/\left(n+1\right),1/n\right],$ $n\in\N$). However, it appears that it is possible to built up an It\^o calculus based on the truncated variation (see \cite[Sect. 4]{LochQVTV:2017}).

An open question remains if the existence of the limit $c\cdot\TTV{\omega}{\left[0,T\right]}c$
as $c\ra0+$ implies the existence of \emph{the same quadratic variation}
along \emph{some} sequence of partitions of $\omega.$ Below we give an example of such sequence for continuous
$\omega$ satisfying some additional condition. This sequence is a combination of both - the sequence of drawup and drawdown times and the sequence of the Lebesgue partitions.
The sequence of the drawup and drawdown times $\rho_k^n,$ $n \in \N,$ $k=0,1,\ldots$, is defined in the following way: let
\[
\rho_u^n = \inf \cbr{t \in (0, T]:  \omega(t) - \min_{0 < s \le t} \omega(s) = 2^{-n}},
\]
\[
\rho_d^n = \inf \cbr{t \in (0, T]: \max_{0 < s \le t} \omega(s) - \omega(t) = 2^{-n}}.
\]
If $\rho_u^n < \rho_d^n$ then we define $\rho_0^n = 0,$ $\rho_1^n = \rho_u^n,$ and for $k=1,2,\ldots$
\[
\rho_{2k}^n = \inf \cbr{t \in (\rho_{2k-1}^n, T]:  \max_{\rho_{2k-1}^n < s \le t} \omega(s) - \omega(t) = 2^{-n}},
\]
\[
\rho_{2k+1}^n = \inf \cbr{t \in (\rho_{2k}^n, T]:  \omega(t) - \min_{\rho_{2k}^n < s \le t} \omega(s) = 2^{-n}};
\]
otherwise we define $\rho_0^n = 0,$ $\rho_1^n = \rho_d^n,$ and for $k=1,2,\ldots$
\[
\rho_{2k}^n = \inf \cbr{t \in (\rho_{2k-1}^n, T]:   \omega(t) - \min_{\rho_{2k-1}^n < s \le t} \omega(s) = 2^{-n}},
\]
\[
\rho_{2k+1}^n = \inf \cbr{t \in (\rho_{2k}^n, T]:  \max_{\rho_{2k}^n < s \le t} \omega(s) - \omega(t) = 2^{-n}}.
\]
Now for $n \in \N$ and $k = 0, 1,2,\ldots,$ we define $\tau_{k,0}^n = \rho_k^n$ and for $i = 1,2, \ldots,$
\[
\tau_{k,i}^n = \inf\cbr{t \in \left(\tau_{k,i-1}^n, \rho_{k+1}^n \right]: \left| \omega\rbr{t} - \omega\rbr{\tau_{k,i-1}^n} \right| = 2^{-n}}.
\]
In all above definitions we assume that for $a \ge b,$ $(a,b] = \emptyset$ and  that $\inf \emptyset = +\ns.$

Finally, for  $n \in \N,$ we define $\tau_j^n,$ $j=0,1,\ldots,$ to be the ordered sequence of all $\tau_{k,i}^n,$ $k,i=0,1,\ldots.$

For $k \in \N$ by $i(k,n)$ let us define the greatest $i$ such that $\tau_{k,i}^n < \rho_{k+1}^n$ if $\rho_{k+1}^n <+ \ns$ and $i(k,n) = 0$ otherwise. Next, let $I_{\cbr{\rho_{k+1}^n \le T}} = 1$ if $\rho_{k+1}^n \le T$ and $I_{\cbr{\rho_{k+1}^n \le T}} = 0$ otherwise. We have the following
\begin{stw}
Let $\omega : \sbr{0,T} \ra \R$ be a continuous path. Assume that
\[
\lim_{n \ra + \ns}2^{-n} \TTVemph{\omega}{[0, T]}{2^{-n}}  = Q
\]
and
\begin{eqnarray}
\lim_{n\ra \ns} \sum_{k=1}^{\ns} \big( 2^{-n} \left| \omega\rbr{\rho_{k+1}^n \wedge T} + (-1)^{k+1}2^{-n}I_{\cbr{\rho_{k+1}^n \le T}} -  \omega\rbr{\tau_{k,i(k,n)}^n\wedge T} \right| \nonumber \\
- \rbr{ \omega\rbr{\rho_{k+1}^n \wedge T}  -  \omega\rbr{\tau_{k,i(k,n)}^n\wedge T} }^2 \big) = 0 \label{cond}
\end{eqnarray}
if $\rho_u^n < \rho_d^n$ or
\begin{eqnarray}
\lim_{n\ra \ns} \sum_{k=1}^{\ns} \big( 2^{-n} \left| \omega\rbr{\rho_{k+1}^n \wedge T} + (-1)^{k}2^{-n}I_{\cbr{\rho_{k+1}^n \le T}} -  \omega\rbr{\tau_{k,i(k,n)}^n\wedge T} \right| \nonumber \\
- \rbr{ \omega\rbr{\rho_{k+1}^n \wedge T}  -  \omega\rbr{\tau_{k,i(k,n)}^n\wedge T} }^2 \big) = 0 \label{cond1}
\end{eqnarray}
if $\rho_u^n \ge \rho_d^n.$
Then also
\[
\lim_{n\ra \ns} \sum_{k=1}^{\ns} \rbr{ \omega\rbr{\tau_{k+1}^n \wedge T}  -  \omega\rbr{\tau_{k}^n\wedge T} }^2 = Q.
\]
\end{stw}
We do not give the proof of this Proposition since it is rather technical. The result follows form the decomposition of the truncated variation between the consecutive drawup and drawdown times, see \cite{LochowskiMilosSPA:2013}.  From results of \cite{LochowskiMilosSPA:2013} it follows that paths of continuous real semimartingales satisfy (\ref{cond}) and (\ref{cond1}) with probability $1$. For a standard Brownian motion it is also possible to prove (\ref{cond}) and (\ref{cond1}) directly, using formulas proven by Taylor in \cite{Taylor:1975kx}.

\subsection*{Acknowledgements}
The work of L. Ch. G. and F. J. M. was supported in part by the National Research Foundation of South Africa (Grant Number: 105924). The work of R. M. {\L}.  was partially funded by the National Science Centre, Poland, under Grant No.~$2016/21/$B/ST$1/01489$. R. M. \L. is very grateful to Prof. Volodya Vovk for very insightful remarks and encouragement. The authors are also grateful to anonymous referee, whose remarks helped to state the obtained results in much more precise way.

\end{document}